\pdfoutput=1
\documentclass[journal]{IEEEtran}

\usepackage{amsmath,amssymb,amsfonts,amsthm}
\usepackage{bm}
\usepackage{booktabs}
\usepackage{cite}
\usepackage{graphicx}
\usepackage[hidelinks]{hyperref}
\usepackage{xcolor}
\usepackage{mathtools}
\usepackage{comment}
\usepackage{csvsimple}
\usepackage{pgfplotstable}
\graphicspath{{./}}
\usepackage{siunitx}
\theoremstyle{plain}
\newtheorem{theorem}{Theorem}
\newtheorem{lemma}{Lemma}

\theoremstyle{definition}
\newtheorem{assumption}{Assumption}
\theoremstyle{remark}
\newtheorem{remark}{Remark}
\newcommand{\R}{\mathbb{R}}

\newcommand{\norm}[1]{\left\lVert #1\right\rVert}

\newcommand{\tht}{\tilde{\theta}}
\newcommand{\ths}{\theta^{*}}
\newcommand{\proj}{\operatorname{proj}}
\newcommand{\lammin}[1]{\lambda_{\min}\left(#1\right)}
\newcommand{\lammax}[1]{\lambda_{\max}\left(#1\right)}
\def\BibTeX{{\rm B\kern-.05em{\sc i\kern-.025em b}\kern-.08em
    T\kern-.1667em\lower.7ex\hbox{E}\kern-.125emX}}

\begin{document}
\pagestyle{empty}

\title{Nesterov-Accelerated Concurrent Learning for Lyapunov-Based Deep Neural Networks}

\author{Hossein Papi$^{1}$, Omkar Sudhir Patil$^{1}$
\thanks{$^{1}$ Hossein Papi and Omkar Sudhir Patil are with the Division of Electrical and Computer Engineering, Louisiana State University, Baton Rouge, LA 70803 USA (e-mail: hpapi1@lsu.edu, opatil1@lsu.edu).}}

\maketitle\thispagestyle{empty}

\begin{abstract}
Concurrent learning (CL) has recently been extended to Lyapunov-based deep
neural networks (DNNs) to obtain parameter convergence under finite rather
than persistent excitation. However, the existing first-order law converges
slowly, especially for high-dimensional systems and overparameterized
networks. To address this problem, this paper develops a
Nesterov-accelerated concurrent learning (NACL) adaptation law for uncertain
second-order control-affine systems. The developed law drives momentum
dynamics with recorded state-input data, so that all hidden layers of the
DNN are updated online under a verifiable finite excitation condition. A
dynamic state-derivative observer is designed to reconstruct the recorded
input without acceleration measurements. The momentum variable introduces a
cross term coupling the momentum lag error to the state-dependent DNN
Jacobian and the tracking error, and a state-dependent normalization is
introduced into the momentum filter to dominate it. A Lyapunov-based
analysis establishes exponential convergence of the tracking, momentum,
observer, and weight estimation errors to a bounded residual ball.
Simulations on a six-degree-of-freedom unmanned underwater vehicle
demonstrate a $24.5\%$ and $62.9\%$ reduction in the root-mean-square
function approximation error relative to first-order CL and to a
higher-order tuner without CL, respectively.
\end{abstract}

\begin{IEEEkeywords}
Adaptive control, deep neural networks, higher order tuner, concurrent learning, Lyapunov methods, nonlinear control systems.
\end{IEEEkeywords}

\section{Introduction}
\IEEEPARstart{D}{eep} neural networks (DNNs) are universal function
approximators \cite{Hornik1989Multilayer, Narendra1990Identification} and
have long been used to compensate for unstructured uncertainty in nonlinear
systems \cite{Lewis1996Neural, Sanner1992Gaussian, Ge2002Stable,
Farrell2006Poly, Lavretsky2012Robust}. Conventional DNN-based controllers are
trained offline on a fixed dataset and held constant during operation, so no
guarantee relates the learned model to the dynamics encountered at run time.
Lyapunov-based DNNs (Lb-DNNs) remove the offline stage by deriving the
adaptation law for every layer from a stability analysis
\cite{patil2022lyapunov, joshi2020design, sun2022lyapunov,
griffis2023lyapunov, hart2024lyapunov, makumi2025lyapunov}, which yields
real-time weight adaptation with a guarantee on the tracking error.

However, a guarantee on the tracking error is not a guarantee on learning
the model. Since the update laws of \cite{patil2022lyapunov,
sun2022lyapunov} are driven by the instantaneous tracking error alone, the
weight estimates need not converge, and the identified DNN cannot be relied
upon for prediction, fault detection, or operation without continuous
feedback \cite{Patil2024Simultaneous}. Parameter convergence classically
requires persistence of excitation (PE) \cite{Ioannou1996Robust,
Narendra1989Stable, Astrom2008Adaptive, Sastry1989Adaptive}. For a DNN the
condition falls on the network Jacobian, which depends on the weight
estimates themselves, so the requirement is trajectory dependent, cannot be
verified online, and forces a continual tradeoff between exploration and
tracking.

Data-driven adaptation replaces PE with initial, interval, or finite
excitation (FE) conditions imposed on recorded signals, through dynamic
regressor extension and mixing \cite{Aranovskiy2017DREM, ortega2021new,
bobtsov2022generation}, matrix regressors \cite{jenkins2020fast}, and
filter-based and integral constructions \cite{kreisselmeier1977adaptive,
roy2018combined, katiyar2023initial, cho2018composite}. Concurrent learning
(CL) \cite{chowdhary2013concurrent, kamalapurkar2017concurrent,
chowdhary2013cladmpc}, integral CL \cite{parikh2019integral}, and composite
learning \cite{slotine1989composite, pan2018composite} obtain parameter
convergence by driving the update with a prediction error formed from a
history stack, so that the excitation condition becomes verifiable online.
Extending CL to a DNN is challenging because of the nonlinear-in-parameter
(NIP) dependence on the inner-layer weights: the prediction error is not
affine in the estimates, and the Jacobian that takes its place evolves with
them. Lb-DNN controllers have been augmented with CL for control-affine and
Euler-Lagrange plants \cite{Basyal2024Augmentation, Basyal2024EulerLagrange,
Basyal2026Realtime}, and the all-layer case was addressed in
\cite{Hart2025Concurrent} by reconstructing the history stack from the
current Jacobian. However, the Jacobian history stack of a deep network is severely
ill-conditioned, since inner-layer directions are excited only through
products of outer-layer weights, so the first-order law of
\cite{Hart2025Concurrent} learns at the rate of its weakest recorded
direction.

The first-order structure of the update law itself has been revisited.
Higher-order tuners (HOTs), motivated by the continuous-time limit of
Nesterov's accelerated gradient method \cite{Nesterov1983Method,
Su2014Differential, gaudio2019connections}, add an auxiliary momentum state
and a state-dependent normalization, and improve the convergence rate for
constant and time-varying regressors \cite{gaudio2021high,
cui2022accelerated, gaudio2022parameter, BasuRoy2026Exponential,
pan2023least, shi2024composite}. Momentum and concurrent
learning were combined for linear-in-the-parameters Euler-Lagrange systems
in \cite{Le2022Integral}, where an integral CL term drives an accelerated
gradient update under finite excitation. The accelerated structure reaches
the DNN setting in \cite{Le2025Accelerated}, where all weights adapt in real
time and asymptotic tracking is established, but without recorded data.

Combining momentum with concurrent learning for a DNN is not a matter of
replacing the update law of \cite{Hart2025Concurrent}. A higher-order law
splits the weight estimation error into two coupled states, the momentum
variable and the lag of the estimates behind it, and a single CL gain
multiplies both. Consequently, the recorded-data term damps one state while
exciting the other and is no longer purely stabilizing: a richer history
stack drives the estimates toward their ideal values and drives the lag
harder at the same time. Moreover, the momentum variable introduces a cross
term among the lag, the state-dependent DNN Jacobian, and the tracking error
that has no first-order counterpart. Higher-order tuners for
linear-in-the-parameters uncertainty dominate the analogous term with a
normalization built from a known regressor, whereas the DNN Jacobian is
evaluated at the current estimates, so the normalization is estimate
dependent and its boundedness must be established rather than assumed.

This paper develops a Nesterov-accelerated concurrent learning (NACL)
adaptation law for Lyapunov-based DNNs, together with the stability analysis
that resolves the aforementioned obstructions. Relative to
\cite{Le2022Integral}, which combined momentum and CL for known regressors,
the history stack here acts through the current-weight Jacobian, and the
Taylor remainder and the normalization must be designed rather than assumed.
The central element is a dual state-dependent normalization matched to the
two obstructions above. Projection of the weight
estimates and the momentum variable ensures that the normalizers, the
Jacobian, and the Taylor remainder remain uniformly bounded, and a dynamic
state-derivative observer reconstructs the recorded control input without
acceleration measurements. A Lyapunov-based analysis shows that the
tracking, momentum, weight estimation, and observer errors converge
exponentially to an ultimate bound under a finite excitation condition on
the recorded Jacobian data, rather than under persistent excitation of the
trajectory. Simulations on a six-degree-of-freedom unmanned underwater
vehicle show that momentum without recorded data performs worse than
first-order CL, whereas the developed law reduces the root-mean-square
function approximation error by $62.9\%$ relative to the momentum tuner and
by $24.5\%$ relative to first-order CL.

\subsection{Notation}\label{sec:Notation}

The sets of real, positive, and nonnegative real numbers are denoted by
$\mathbb{R}$, $\mathbb{R}_{>0}$, and $\mathbb{R}_{\geq0}$, respectively.
The identity matrix is denoted by $I_n$, and $A^\top$ denotes the transpose
of $A$. The symbols $\|\cdot\|$, $\operatorname{tr}(\cdot)$, and
$\lambda_{\min}(\cdot)$ denote the Euclidean/Frobenius norm, trace, and
minimum eigenvalue, respectively. The space of essentially bounded signals is denoted by $\mathcal{L}_\infty$.

\section{Problem Formulation and Control Design}\label{sec:problem}
\subsection{System Model}
Consider a class of second-order control-affine nonlinear dynamical systems governed by
\begin{equation}
\ddot{x} = f(X) + u,
\label{eq:sys_dyn}
\end{equation}
where $x \in \mathbb{R}^n$ denotes the system state vector, $X = [x^\top, \dot{x}^\top]^\top \in \mathbb{R}^{2n}$ is the composite state vector, $u \in \mathbb{R}^n$ is the control input, and $f(X) \in \mathbb{R}^n$ is an unknown, smooth nonlinear drift vector field.

Let $x_d(t)\in\mathbb{R}^n$ be a reference trajectory that is twice
continuously differentiable and satisfies
$\|x_d\|\le\bar{x}_d$, $\|\dot{x}_d\|\le\bar{\dot{x}}_d$, and
$\|\ddot{x}_d\|\le\bar{\ddot{x}}_d$, where
$\bar{x}_d,\bar{\dot{x}}_d,\bar{\ddot{x}}_d\in\mathbb{R}_{>0}$ are
known constants. Define the tracking error $e \triangleq x - x_d$ and the filtered tracking error $r \in \mathbb{R}^n$ \cite{Ioannou1996Robust, Lewis1996Neural} as
\begin{equation}
r = \dot{e} + \alpha_1 e,
\label{eq:filtered_error}
\end{equation}
where $\alpha_1 > 0$ is a gain constant. The unknown drift function $f(X)$ is modeled using a DNN over a compact domain $\Omega_X \subset \mathbb{R}^{2n}$ \cite{Hornik1989Multilayer, Lewis1996Neural} as
\begin{equation}
f(X) = \Psi(X, \theta^*) + \varepsilon(X),
\label{eq:dnn_approx}
\end{equation}
where $\Psi(X, \theta^*) \in \mathbb{R}^n$ represents the ideal DNN output, $\theta^* \in \mathbb{R}^{n_\theta}$ is the ideal weight vector containing all stacked inner and outer layer weights and biases, and $\varepsilon(X) \in \mathbb{R}^n$ is the function approximation error bounded by $\|\varepsilon(X)\| \le \bar{\varepsilon}$. 
Taking the time derivative of \eqref{eq:filtered_error} and applying
\eqref{eq:sys_dyn} and \eqref{eq:dnn_approx} yields the open-loop
error system
\begin{equation}
\dot{r} = \Psi(X,\theta^*) + \varepsilon(X) + u - \ddot{x}_d
          + \alpha_1\dot{e}.
\label{eq:r_dot_open}
\end{equation}

Since a DNN is NIP with respect to its hidden-layer weights, the function
approximation error cannot be expressed as a linear function of the weight
estimation error. Let $\hat{\theta} \in \mathbb{R}^{n_\theta}$ denote the
weight estimate, and define the parameter estimation error as
$\tilde{\theta} \triangleq \theta^* - \hat{\theta}$. 
The ideal and estimated weights are constrained to the known compact
convex set
$\Theta\triangleq\{\theta\in\mathbb{R}^{n_\theta}:\|\theta\|\le\bar{\theta}\}$,
where $\bar{\theta}\in\mathbb{R}_{>0}$ is a user-selected bound
satisfying $\theta^*\in\Theta$.

\subsection{Closed-Loop Error System and Control Law Development}

The controller is the standard Lb-DNN controller of \cite{patil2022lyapunov},
held fixed so that the contribution lies entirely in the adaptation law of
Section~\ref{sec:adaptation} and so that every CL method compared in
Section~\ref{sec:sim} drives the same closed loop. Since $f(X)$ is unknown,
the control input is designed to cancel the known terms in
\eqref{eq:r_dot_open} and to leave the parametric error
$\mathcal{J}(X,\hat{\theta})\tilde{\theta}$, which is the signal the
adaptation law acts on.
Based on the subsequent stability analysis, the control input is designed as
\begin{equation}
u = \ddot{x}_d - \Psi(X,\hat{\theta}) - k_1 r - e - \alpha_1\dot{e},
\label{eq:control_input}
\end{equation}
where $k_1\in\mathbb{R}_{>0}$ is a user-defined control gain.

Taking a first-order Taylor series expansion of $\Psi(X, \theta^*)$ around $\hat{\theta}$ yields
\begin{equation}
\Psi(X, \theta^*) - \Psi(X, \hat{\theta}) = \mathcal{J}(X, \hat{\theta})\tilde{\theta} + R(X, \tilde{\theta}),
\label{eq:taylor_expansion}
\end{equation}
where $\mathcal{J}(X, \hat{\theta}) \triangleq \left.\frac{\partial \Psi(X, \theta)}{\partial \theta}\right|_{\theta = \hat{\theta}} \in \mathbb{R}^{n \times {n_\theta}}$ is the state-dependent DNN Jacobian matrix, and $R(X, \tilde{\theta}) \in \mathbb{R}^n$ represents the higher-order Taylor remainder term, which satisfies $\|R(X, \tilde{\theta})\| \le \rho_0(\|X\|)\,\|\tilde{\theta}\|^{2}$ on a compact domain for a known polynomial $\rho_0$ \cite{Patil2026Bounds}; the explicit form of $\rho_0$ is given in \eqref{eq:R_bound}.

Substituting \eqref{eq:dnn_approx} and \eqref{eq:control_input} into
the open-loop error system, applying \eqref{eq:taylor_expansion}, and
canceling cross terms yields
\begin{equation}
\dot{r} = \mathcal{J}(X,\hat{\theta})\tilde{\theta}
        + R(X,\tilde{\theta}) + \varepsilon(X) - k_1 r - e.
\label{eq:r_dot}
\end{equation}

\subsection{Dynamic State-Derivative Observer}

To incorporate previously recorded data into the update law in the
concurrent learning style, the update law is augmented with a history stack
containing the error between the applied control input and a reconstructed
version of it. Reconstructing the input from \eqref{eq:sys_dyn} requires
$\ddot{x}$, which is not measured, so a dynamic state-derivative observer
is developed. The observer is designed as
\begin{align}
\dot{\hat{r}}      &= \hat{\Delta} - \ddot{x}_d + \alpha_1\dot{e}
                      + \alpha_2\tilde{r},
\label{eq:obs_r}\\
\dot{\hat{\Delta}} &= \tilde{r} + k_\Delta\tilde{\Delta} + \dot{u},
\label{eq:obs_delta}
\end{align}
where $\hat{r},\hat{\Delta}\in\mathbb{R}^n$ denote the observer
estimates of $r$ and $\ddot{x}$, the observer errors are
$\tilde{r} \triangleq r - \hat{r}$ and
$\tilde{\Delta} \triangleq \ddot{x} - \hat{\Delta}$,
$\alpha_2, k_\Delta\in\mathbb{R}_{>0}$ are constant observer gains,
and $\dot{u}$ is the derivative of the control input. 
Taking the time derivative of $\tilde{r}$ and $\tilde{\Delta}$ and
applying \eqref{eq:sys_dyn}, \eqref{eq:obs_r}, and
\eqref{eq:obs_delta} yields
\begin{align}
\dot{\tilde{r}}      &= \tilde{\Delta} - \alpha_2\tilde{r},
\label{eq:obs_err_r}\\
\dot{\tilde{\Delta}} &= \dot{f}(X) - \tilde{r} - k_\Delta\tilde{\Delta},
\label{eq:obs_err_delta}
\end{align}
where $\dot{f}(X)\triangleq\frac{\partial f}{\partial x}\dot{x}
+ \frac{\partial f}{\partial\dot{x}}\ddot{x}$.

The observer error $\tilde{r}$ is known, since $r$ and $\hat{r}$ are
both known. Because \eqref{eq:obs_err_r} gives
$\tilde{\Delta} = \dot{\tilde{r}} + \alpha_2\tilde{r}$, an
implementable form of \eqref{eq:obs_delta} is obtained by integrating
both sides and using
$\int_{t_0}^{t}\dot{\tilde{r}}\,d\tau = \tilde{r}(t)-\tilde{r}(t_0)$,
which yields
\begin{align}
\hat{\Delta}(t) ={}& \hat{\Delta}(t_0)
   + k_\Delta\big(\tilde{r}(t)-\tilde{r}(t_0)\big)
   + \big(u(t)-u(t_0)\big)\notag\\
&+ \int_{t_0}^{t}\big(k_\Delta\alpha_2+1\big)\tilde{r}(\tau)\,d\tau,
\label{eq:obs_implementable}
\end{align}
so that $\dot{u}$ is never required online.

Let the concatenated observer error vector
$z\in\mathbb{R}^{2n}$ be defined as
$z\triangleq[\tilde{r}^{\top},\tilde{\Delta}^{\top}]^{\top}$, and let
the composite error vector be
$\zeta\triangleq[r^{\top},e^{\top},\vartheta^{\top},\eta^{\top}]^{\top}
\in\mathbb{R}^{\psi}$ with $\psi\triangleq 2n+2n_\theta$, where the
parameter coordinates $\vartheta,\eta\in\mathbb{R}^{n_\theta}$ are
defined in \eqref{eq:hot_coords} below. Let the open
and connected set $\mathcal{D}\subset\mathbb{R}^{\psi}$ be defined as
$\mathcal{D}\triangleq\{\zeta\in\mathbb{R}^{\psi}:\|\zeta\|<\chi\}$,
where $\chi\in\mathbb{R}_{>0}$ is a subsequently defined known upper
bound. The following lemma bounds $\dot{f}$ on $\mathcal{D}$.

\begin{lemma}\label{lem:fdot}
For all $\zeta\in\mathcal{D}$, there exists a constant
$\delta_f\in\mathbb{R}_{>0}$ such that
$\|\dot{f}(X)\|\le\delta_f$.
\end{lemma}
\begin{proof}
See \cite{Hart2025Concurrent}.
\end{proof}

Let $\Lambda_0\triangleq\min\{k_\Delta,\,2\alpha_2\}$, and let the initial
acceleration error be bounded by a known constant
$\|\tilde{\Delta}(t_0)\|\le\bar{\Delta}_0$. The following lemma summarizes
the property of the observer used in the subsequent development.

\begin{lemma}\label{lem:observer}
Consider the observer given by \eqref{eq:obs_r} and \eqref{eq:obs_delta}
with $\hat{r}(t_0)=r(t_0)$, and let $\delta_\Delta\in\mathbb{R}_{>0}$ be a
prescribed accuracy satisfying the feasibility gain condition
\begin{equation}
k_\Delta\Lambda_0 > \frac{\delta_f^2}{\delta_\Delta^2}.
\label{eq:feasibility}
\end{equation}
Provided $\zeta(\tau)\in\mathcal{D}$ for all $\tau\in[t_0,t]$, the observer
error satisfies $\|z(t)\|\le\delta_\Delta$ for all $t\ge t_\Delta$, where
the settling time $t_\Delta\ge t_0$ is a known function of $\Lambda_0$,
$k_\Delta$, $\delta_f$, $\delta_\Delta$, and $\bar{\Delta}_0$.
\end{lemma}
\begin{proof}
See \cite{Hart2025Concurrent}.
\end{proof}

\section{NACL Adaptation Law}
\label{sec:adaptation}
 
The adaptation law is developed by extending CL to the NIP structure of the
DNN and augmenting the resulting update with higher-order momentum dynamics.
The shorthand
$\mathcal{J}(X,\hat{\theta})$ denotes the DNN Jacobian at the current state and
weight estimates, and $\mathcal{J}_i\triangleq\mathcal{J}(X_i,\hat{\theta})$
denotes the Jacobian evaluated using the current weight estimates and
the previously recorded state $X_i\triangleq X(t_i)$. Compared to CL
approaches that begin gathering data immediately, the history stack
$\mathcal{H}=\{(X_i,u_i,\hat{\Delta}_i)\}_{i=1}^{N}$, with
$\hat{\Delta}_i$ the observer estimate at the recording time, is
constructed only for
$t_i\ge t_\Delta$, where the observer settling time $t_\Delta$ is
characterized in Lemma~\ref{lem:observer}. Using the observer estimate
$\hat{\Delta}_i\triangleq\hat{\Delta}(t_i)$, the control input at each
recorded point is reconstructed as
\begin{equation}
\hat{u}_i = \hat{\Delta}_i - \Psi(X_i,\hat{\theta}).
\label{eq:u_hat_i}
\end{equation}
 
The implementable form of the NACL update law is designed as
\begin{align}
\dot{\Lambda} &= \proj\left\{\Gamma\left[\mathcal{J}^{\top}(X,\hat{\theta})r
   -\frac{\gamma_1}{N}\sum_{i=1}^{N}\mathcal{J}_i^{\top}\big(u_i-\hat{u}_i\big)
   -\gamma_2\hat{\theta}\right]\right\},
\label{eq:lambda_dot}\\
\dot{\hat{\theta}} &= \proj\Big\{-\beta\,\Gamma\big(\hat{\theta}-\Lambda\big)
   \mathcal{N}(t)\Big\},
\label{eq:theta_hat_dot}
\end{align}
where $\Lambda\in\R^{n_\theta}$ is an auxiliary momentum variable,
$\beta,\gamma_1,\gamma_2\in\R_{>0}$ are user-selected adaptation
gains, and $\Gamma(t)\in\R^{n_\theta\times n_\theta}$ is a
positive-definite time-varying least-squares adaptation gain. The update in
\eqref{eq:lambda_dot} performs the gradient-like adjustment using the
instantaneous tracking error and the recorded data, whereas
\eqref{eq:theta_hat_dot} drives the weights toward $\Lambda$ through a
normalized first-order filter; this is the Nesterov-inspired momentum
structure. The projection operator confines both $\Lambda$ and
$\hat{\theta}$ to $\Theta$, which renders the DNN bounds used in
Lemma~\ref{lem:fdot} and in Section~\ref{sec:stability} uniform and yields
$\hat{\theta}\in\mathcal{L}_\infty$.

The projection operator satisfies the property~\cite[App.~E]{Krstic1995} that, for the projected estimate
$\psi\in\{\Lambda,\hat{\theta}\}$, any $\xi\in\Theta$, and any
$y\in\mathbb{R}^{n_\theta}$,
\begin{equation}
(\psi-\xi)^{\top}\Gamma^{-1}
   \big(\operatorname{proj}\{\Gamma y\}-\Gamma y\big)\le 0 .
\label{eq:proj_property}
\end{equation}
 
The state-dependent normalization signal in \eqref{eq:theta_hat_dot}
is defined as
\begin{equation}
\mathcal{N}(t) \triangleq 1+\mu_1\norm{\mathcal{J}(X,\hat{\theta})}_{2}^{2}
   +\mu_2\operatorname{tr}\big(\Phi_{\mathcal{J}}\big),
\label{eq:norm_factor}
\end{equation}
where $\norm{\cdot}_2$ is the induced $2$-norm, $\mu_1,\mu_2\in\R_{>0}$,
and $\Phi_{\mathcal{J}}$ is the averaged Jacobian history stack
\begin{equation}
\Phi_{\mathcal{J}}(t) \triangleq
   \frac{1}{N}\sum_{i=1}^{N}\mathcal{J}^{\top}(X_i,\hat{\theta})\mathcal{J}(X_i,\hat{\theta}).
\label{eq:phi_j}
\end{equation}
The two state-dependent terms in \eqref{eq:norm_factor} are what allow the
higher-order cross term to be dominated, as made precise by the gain
conditions in Section~\ref{sec:stability}.
 
The adaptation gain evolves according to the guarded least-squares law
with forgetting \cite{Hart2025Concurrent}
\begin{equation}
\frac{d}{dt}\Gamma^{-1}(t)=
\begin{cases}
-\kappa\,\Gamma^{-1}(t)+\nu\gamma_1\Phi_{\mathcal{J}}(t), 
 & \text{if }\lambda_{\Gamma,\min}<\lammin{\Gamma}\\[-3pt]
&\hspace{-3em}\text{and }\lammax{\Gamma}<\lambda_{\Gamma,\max},\\[5pt]
0, & \text{otherwise},
\end{cases}
\label{eq:gain_law}
\end{equation}
where $\kappa\in\R_{>0}$ is a forgetting factor, $0<\nu<2$, and
$\lambda_{\Gamma,\min},\lambda_{\Gamma,\max}\in\R_{>0}$ are
user-selected bounds on the spectrum of $\Gamma$ with
$\Gamma(t_0)$ initialized so that
$\lambda_{\Gamma,\min}I\preceq\Gamma(t_0)\preceq\lambda_{\Gamma,\max}I$.
Freezing the update outside the band prevents covariance windup during
intervals of low excitation, where $\Phi_{\mathcal{J}}$ is rank deficient
and the forgetting term alone would drive $\Gamma^{-1}$ toward zero, and it
ensures
\begin{equation}
0<\lambda_{\Gamma,\min}I\preceq\Gamma(t)\preceq\lambda_{\Gamma,\max}I,
\qquad \forall\,t\ge t_0,
\label{eq:gain_band}
\end{equation}
by construction. In implementation, the two terms are faded continuously
rather than switched: the forgetting term is scaled by
$\max\{0,\,1-\lambda_{\max}(\Gamma)/\lambda_{\Gamma,\max}\}$ and the
history-stack term by
$\min\{1,\,\max\{0,\,\lambda_{\min}(\Gamma)/\lambda_{\Gamma,\min}-1\}\}$;
the subsequent analysis covers this variant, since it only requires the
branch weights to lie in $[0,1]$. Since $\nu$ and $\gamma_1$ enter
\eqref{eq:gain_law} only through the product $\nu\gamma_1$, the bound
on $\nu$ places no restriction on the CL gain; the choice $\nu=1$
recovers the covariance law used in \cite{Hart2025Concurrent}. For
notational convenience let $s(t)\in\{0,1\}$ denote the indicator of
the active branch of \eqref{eq:gain_law}, so that
$\frac{d}{dt}\Gamma^{-1}=s(t)\big[-\kappa\Gamma^{-1}
+\nu\gamma_1\Phi_{\mathcal{J}}\big]$. All subsequent identities hold
almost everywhere, with $\Gamma$ continuous and without gain resets.
 
An analytical form of \eqref{eq:lambda_dot} is obtained from the
observer construction. Using \eqref{eq:u_hat_i} with the DNN model in
\eqref{eq:dnn_approx} and the Taylor expansion in
\eqref{eq:taylor_expansion}, the reconstruction error at each recorded
point satisfies $u_i-\hat{u}_i=\tilde{\Delta}_i-\mathcal{J}_i\tht
-R_i-\varepsilon_i$. Multiplying by $-\gamma_1\mathcal{J}_i^{\top}/N$
and summing, and noting that $\tht$ carries no index and therefore
factors out of the sum, yields
\begin{equation}
-\frac{\gamma_1}{N}\sum_{i=1}^{N}\mathcal{J}_i^{\top}\big(u_i-\hat{u}_i\big)
 = \gamma_1\Phi_{\mathcal{J}}\tht+\gamma_1 d_R+\gamma_1 d_\varepsilon,
\label{eq:cl_decomp}
\end{equation}
where $d_R\triangleq\frac{1}{N}\sum_{i=1}^{N}\mathcal{J}_i^{\top}R(X_i,\tht)$
and $d_\varepsilon\triangleq\frac{1}{N}\sum_{i=1}^{N}\mathcal{J}_i^{\top}
(\varepsilon_i-\tilde{\Delta}_i)$. Consequently
$\dot{\Lambda}=\proj\{\Gamma G\}$ with
\begin{equation}
G \triangleq \mathcal{J}^{\top}(X,\hat{\theta})r
  +\gamma_1\Phi_{\mathcal{J}}\tht
  +\gamma_1 d_R+\gamma_1 d_\varepsilon-\gamma_2\hat{\theta}.
\label{eq:G_def}
\end{equation}
Equation \eqref{eq:cl_decomp} is the step that makes CL applicable to a DNN:
the recorded-data term resolves into a positive semidefinite history stack acting
on $\tht$, plus two disturbances arising from the Taylor remainder and from
the approximation and observer errors, respectively.
 
Because the update evolves in two coupled parameter states, the
analysis is performed in the coordinates
\begin{equation}
\eta \triangleq \hat{\theta}-\Lambda,
\qquad
\vartheta \triangleq \ths-\Lambda,
\label{eq:hot_coords}
\end{equation}
so that $\tht=\vartheta-\eta$. Here $\eta$ quantifies the lag of the
weight estimates behind the momentum variable and $\vartheta$ the
distance of the momentum variable from the ideal weights. Since
$\ths$ is constant, \eqref{eq:lambda_dot} and
\eqref{eq:theta_hat_dot} give
\begin{equation}
\dot{\vartheta} = -\proj\{\Gamma G\},
\qquad
\dot{\eta} = \proj\big\{-\beta\Gamma\eta\mathcal{N}\big\}
            -\proj\{\Gamma G\}.
\label{eq:coord_dyn}
\end{equation}
 
\begin{assumption}[Finite Excitation]\label{ass:fe}
There exist $\lambda_e\in\R_{>0}$ and $\mathcal{T}>t_\Delta$ such that
the current-weight CL history stack satisfies
\begin{equation}
\lammin{\Phi_{\mathcal{J}}(t)}\ge\lambda_e,
\qquad \forall\, t\ge\mathcal{T}.
\label{eq:fe_assumption}
\end{equation}
\end{assumption}
 
Assumption~\ref{ass:fe} is imposed on the retained samples $\{X_i\}$,
which are recorded over a finite interval and re-evaluated through
$\mathcal{J}(X_i,\hat{\theta})$ at the current weights, in contrast to a
PE condition, which would require excitation of the closed-loop
trajectory for all time.

\begin{lemma}\label{lem:data}
Let $t'\ge t_\Delta$ and suppose that $\zeta(t)\in\mathcal{D}$ for all
$t\in[t_0,t']$. Then every sample retained at a time
$t_i\in[t_\Delta,t']$ satisfies
\begin{equation}
\norm{\tilde{\Delta}(t_i)}\le\delta_\Delta .
\label{eq:data_accuracy}
\end{equation}
\end{lemma}
\begin{proof}
Since $\zeta(t)\in\mathcal{D}$ on $[t_0,t']$, Lemma~\ref{lem:fdot}
bounds $\dot{f}$ by $\delta_f$ on that interval, so Lemma~\ref{lem:observer}
applies and yields $\norm{z(t)}\le\delta_\Delta$ for all
$t\in[t_\Delta,t']$. The claim follows from
$\norm{\tilde{\Delta}(t_i)}\le\norm{z(t_i)}$ and the retention rule
$t_i\ge t_\Delta$.
\end{proof}

Theorem~\ref{thm:main} establishes the invariance of $\mathcal{D}$ from
$t_0$ by a first-exit argument, so no separate assumption on the recorded
data is required; the retention rule $t_i\ge t_\Delta$ is the only condition
enforced in implementation.

\section{Stability Analysis}\label{sec:stability}
 
Using the structure of the DNN, the following bounds are established.
By \cite[Lemma~2]{Patil2026Bounds}, the Jacobian is bounded on
$\Theta$ and the corresponding compact input domain by a constant
$\bar{\mathcal{J}}\in\R_{>0}$, so that
$\norm{\mathcal{J}(X,\hat{\theta})}\le\bar{\mathcal{J}}$ and
$\norm{\mathcal{J}_i}\le\bar{\mathcal{J}}$. By
\cite[Theorem~1]{Patil2026Bounds}, which follows from the DNN Hessian
bound therein, the Lagrange remainder in
\eqref{eq:taylor_expansion} satisfies
\begin{align}
\norm{R(X,\tht)}\le\rho_0(\norm{X})\norm{\tht}^{2},\label{eq:R_bound_a}\\
\quad
\rho_0(\norm{X})=a_2\norm{X}^{2}+a_1\norm{X}+a_0,
\label{eq:R_bound}
\end{align}
with $a_2,a_1,a_0\in\R_{>0}$ computable from the layer widths, the
activation constants, and $\bar{\theta}$. Because $\ths,\hat{\theta}\in\Theta$
and $\Theta$ is convex, the Taylor segment from $\hat{\theta}$ to $\ths$ lies
in $\Theta$. On $\mathcal{D}$, the state bounds established in the proof of
Lemma~\ref{lem:fdot} give $\norm{X}\le c_0+c_1\norm{\zeta}$ with
$c_0=\bar{x}_d+\bar{\dot{x}}_d$ and $c_1=2+\alpha_1$, so that
\eqref{eq:R_bound_a} and \eqref{eq:R_bound} yield
\begin{equation}
\norm{R}\le\bar{\rho}_0(\chi)\norm{\tht}^{2},
\ \
\bar{\rho}_0(\chi)\triangleq a_2(c_0+c_1\chi)^{2}
   +a_1(c_0+c_1\chi)+a_0 .
\label{eq:R_bound_domain}
\end{equation}
Whenever every retained sample was recorded while $\zeta\in\mathcal{D}$,
so that $\norm{X_i}\le c_0+c_1\chi$ and, by Lemma~\ref{lem:data},
\eqref{eq:data_accuracy} holds, the averaged stack disturbances
satisfy
\begin{align}
\norm{d_R}&\le C_R\norm{\tht}^{2},
&& C_R\triangleq \bar{\mathcal{J}}\bar{\rho}_0(\chi),
\label{eq:dR_bound}\\
\norm{d_\varepsilon}&\le\bar{d}_\Phi,
&& \bar{d}_\Phi\triangleq \bar{\mathcal{J}}
   \big(\bar{\varepsilon}+\delta_\Delta\big).
\label{eq:deps_bound}
\end{align}

Consider the candidate Lyapunov function
$V:\R^{\psi}\times\R_{\ge0}\to\R_{\ge0}$ defined as
\begin{equation}
V(\zeta,t)\triangleq\tfrac12 r^{\top}r+\tfrac12 e^{\top}e
 +\tfrac12\vartheta^{\top}\Gamma^{-1}\vartheta
 +\tfrac12\eta^{\top}\Gamma^{-1}\eta .
\label{eq:lyap}
\end{equation}
Because $\Gamma$ is time-varying, converting a bound on $V$ into a
bound on $\zeta$ requires uniform spectral bounds on the gain. These
are supplied directly by the guard in \eqref{eq:gain_law}, which
enforces \eqref{eq:gain_band}, whence
\begin{equation}
\beta_1\norm{\zeta}^{2}\le V(\zeta,t)\le\beta_2\norm{\zeta}^{2},
\label{eq:lyap_sandwich}
\end{equation}
with $\beta_1\triangleq\frac12\min\{1,\lambda_{\Gamma,\max}^{-1}\}$
and $\beta_2\triangleq\frac12\max\{1,\lambda_{\Gamma,\min}^{-1}\}$.
 
Let $c_\nu\triangleq1-\nu/2>0$ and $d_\nu\triangleq1+\nu/2$, and
select the gains to satisfy
\begin{equation}
\mu_1\ge\frac{2}{\beta k_1},
\qquad
\mu_2\ge\frac{d_\nu\gamma_1}{\beta},
\qquad
\beta>\gamma_2 .
\label{eq:gain_conditions}
\end{equation}
Define $A_0\triangleq\gamma_2$,
$A\triangleq c_\nu\gamma_1\lambda_e+\gamma_2$,
$B\triangleq\beta-\gamma_2$,
$D_0\triangleq\gamma_1\bar{d}_\Phi+\gamma_2\bar{\theta}$, and, for
$a\in\{A_0,A\}$,
\begin{align}
\rho_q(a)&\triangleq\frac{2\bar{\rho}_0(\chi)^{2}}{k_1}
  +\gamma_1^{2}C_R^{2}\left(\frac{1}{a}+\frac{1}{B}\right),
\label{eq:rho_q}\\
\iota(a)&\triangleq\frac{2\bar{\varepsilon}^{2}}{k_1}
  +D_0^{2}\left(\frac{1}{a}+\frac{1}{B}\right),
\label{eq:iota}
\end{align}
together with
$\lambda(a)\triangleq\min\{k_1/4,\ \alpha_1,\ a/2,\ B/2\}$. The
constants indexed by $A_0$ govern the interval before the excitation
condition is met, when the $\vartheta$ channel is damped by the leakage
gain alone, and those indexed by $A$ govern the interval after it.
Since $A\ge A_0$, it follows that $\rho_q(A)\le\rho_q(A_0)$,
$\iota(A)\le\iota(A_0)$, and $\lambda(A)\ge\lambda(A_0)$. For a
selected convergence rate $0<\lambda_d<\lambda(A_0)$, let the domain
radius satisfy $4\rho_q(A_0)\chi^{2}\le\lambda(A_0)-\lambda_d$, define
$\lambda_{\mathcal{T}}\triangleq\lambda(A)-4\rho_q(A)\chi^{2}
\ge\lambda_d$, and let the set of stabilizing initial conditions be
\begin{equation}
\mathcal{S}\triangleq\Big\{\zeta\in\R^{\psi}:
  \norm{\zeta}<\sqrt{\beta_1/\beta_2}\,\chi\Big\}.
\label{eq:S_set}
\end{equation}
 
\begin{theorem}\label{thm:main}
Consider the system in \eqref{eq:sys_dyn} with the control input
\eqref{eq:control_input}, the observer in \eqref{eq:obs_r} and
\eqref{eq:obs_delta}, and the adaptation laws
\eqref{eq:lambda_dot} and \eqref{eq:theta_hat_dot} with the gain law
\eqref{eq:gain_law}. Let the feasibility condition
\eqref{eq:feasibility} and the gain conditions
\eqref{eq:gain_conditions} be satisfied, let Assumption~\ref{ass:fe}
hold, let $\|\tilde{\Delta}(t_0)\|\le\bar{\Delta}_0$, and let
$\zeta(t_0)\in\mathcal{S}$ with
$\beta_2\iota(A_0)/\lambda_d<\beta_1\chi^{2}$. Then
$\zeta(t)\in\mathcal{D}$ for all $t\ge t_0$, every retained sample
satisfies \eqref{eq:data_accuracy}, the observer error satisfies
$\norm{z(t)}\le\delta_\Delta$ for all $t\ge t_\Delta$, and the composite error
vector is bounded in the sense that
\begin{align}
\norm{\zeta(t)}&\le
 \sqrt{\frac{\beta_2}{\beta_1}\,\norm{\zeta(t_0)}^{2}\,
       e^{-\frac{\lambda_d}{\beta_2}(t-t_0)}
      +\frac{\beta_2\,\iota(A_0)}{\beta_1\lambda_d}},
\quad t\ge t_0,
\label{eq:main_bound_pre}\\
\norm{\zeta(t)}&\le
 \sqrt{\frac{\beta_2}{\beta_1}\,\norm{\zeta(\mathcal{T})}^{2}\,
       e^{-\frac{\lambda_{\mathcal{T}}}{\beta_2}(t-\mathcal{T})}
      +\frac{\beta_2\,\iota(A)}{\beta_1\lambda_{\mathcal{T}}}},
\quad t\ge\mathcal{T}.
\label{eq:main_bound}
\end{align}
\end{theorem}
 
\begin{proof}
Let $t_e\triangleq\sup\{t\ge t_0:\zeta(\tau)\in\mathcal{D}
\ \forall\,\tau\in[t_0,t]\}$. Since $\zeta(t_0)\in\mathcal{S}
\subset\mathcal{D}$ and $\zeta$ is continuous, $t_e>t_0$. Throughout
$[t_0,t_e)$, $\zeta\in\mathcal{D}$, so by Lemma~\ref{lem:data} every
sample retained before $t_e$ satisfies \eqref{eq:data_accuracy}, and
the bounds \eqref{eq:R_bound_domain}--\eqref{eq:deps_bound} hold. The
derivative of $V$ is now bounded on $[t_0,t_e)$.
Let $P\triangleq\Gamma^{-1}$ and
$V_\theta\triangleq\tfrac12\vartheta^{\top}P\vartheta
+\tfrac12\eta^{\top}P\eta$, so that $V=\tfrac12\norm{r}^{2}
+\tfrac12\norm{e}^{2}+V_\theta$. Since $P$ is symmetric and
time-varying, the product rule applied to \eqref{eq:lyap} yields
\begin{align}
\dot{V}={}&r^{\top}\dot{r}+e^{\top}\dot{e}
 +\underbrace{\vartheta^{\top}P\dot{\vartheta}
              +\eta^{\top}P\dot{\eta}}_{\text{adaptation}}
 +\underbrace{\tfrac12\vartheta^{\top}\dot{P}\vartheta
              +\tfrac12\eta^{\top}\dot{P}\eta}_{\text{gain variation}}.
\label{eq:vdot_exact}
\end{align}
The four groups in \eqref{eq:vdot_exact} are evaluated in turn.

Substituting \eqref{eq:r_dot}, $\dot{e}=r-\alpha_1 e$, the scalars $-r^{\top}e$ and $e^{\top}r$
cancel, leaving
\begin{equation}
r^{\top}\dot{r}+e^{\top}\dot{e}
 = r^{\top}\mathcal{J}\tht+r^{\top}R+r^{\top}\varepsilon
   -k_1\norm{r}^{2}-\alpha_1\norm{e}^{2}.
\label{eq:vdot_track}
\end{equation}

Applying the projection property
\eqref{eq:proj_property} to each contribution in
\eqref{eq:coord_dyn}, with $\xi=\ths$ for the $\vartheta$ channel and
with $\xi=\Lambda$ and $\xi=\hat{\theta}$ for the two contributions to
the $\eta$ channel, all of which lie in $\Theta$
\begin{equation}
\vartheta^{\top}P\dot{\vartheta}+\eta^{\top}P\dot{\eta}
 \le-\beta\mathcal{N}\norm{\eta}^{2}-(\vartheta+\eta)^{\top}G .
\label{eq:vdot_adapt}
\end{equation}

Substituting \eqref{eq:gain_law}
\begin{equation}
\tfrac12\vartheta^{\top}\dot{P}\vartheta
+\tfrac12\eta^{\top}\dot{P}\eta
 = s(t)\Big[-\kappa V_\theta+\tfrac{\nu\gamma_1}{2}
   \big(\vartheta^{\top}\Phi_{\mathcal{J}}\vartheta
        +\eta^{\top}\Phi_{\mathcal{J}}\eta\big)\Big].
\label{eq:vdot_gain}
\end{equation}

Substituting \eqref{eq:vdot_track}--\eqref{eq:vdot_gain} into
\eqref{eq:vdot_exact} gives the first bound
\begin{align}
\dot{V}\le{}&-k_1\norm{r}^{2}-\alpha_1\norm{e}^{2}
  -\beta\mathcal{N}\norm{\eta}^{2}
  +r^{\top}\mathcal{J}\tht+r^{\top}R+r^{\top}\varepsilon\notag\\
&-(\vartheta+\eta)^{\top}G
 +s(t)\Big[-\kappa V_\theta+\tfrac{\nu\gamma_1}{2}
   \big(\vartheta^{\top}\Phi_{\mathcal{J}}\vartheta
        +\eta^{\top}\Phi_{\mathcal{J}}\eta\big)\Big].
\label{eq:vdot_1}
\end{align}

Using $\vartheta+\eta=\tht+2\eta$ and expanding $G$ from
\eqref{eq:G_def}, the term $-\tht^{\top}\mathcal{J}^{\top}r$ cancels
$r^{\top}\mathcal{J}\tht$ in \eqref{eq:vdot_1}, since both are equal
scalars, and the residue is the higher-order cross term
$-2\eta^{\top}\mathcal{J}^{\top}r$. Using the identity
$\tht^{\top}\Phi_{\mathcal{J}}\tht
+2\eta^{\top}\Phi_{\mathcal{J}}\tht
=\vartheta^{\top}\Phi_{\mathcal{J}}\vartheta
-\eta^{\top}\Phi_{\mathcal{J}}\eta$, the concurrent-learning
quadratics in $-(\vartheta+\eta)^{\top}G$ and the gain-variation terms
in \eqref{eq:vdot_1} collect into the single quantity
\begin{equation}
\Xi_s\triangleq-\gamma_1\big(1-\tfrac{s\nu}{2}\big)
   \vartheta^{\top}\Phi_{\mathcal{J}}\vartheta
+\gamma_1\big(1+\tfrac{s\nu}{2}\big)
   \eta^{\top}\Phi_{\mathcal{J}}\eta-s\kappa V_\theta ,
\label{eq:branch_terms}
\end{equation}
which is bounded on both branches of \eqref{eq:gain_law} by a single
$s$-free expression: since $0<\nu<2$ gives $c_\nu<1<d_\nu$, since
$\Phi_{\mathcal{J}}\succeq0$, and since $-s\kappa V_\theta\le0$,
\begin{equation}
\Xi_s\le
-c_\nu\gamma_1\vartheta^{\top}\Phi_{\mathcal{J}}\vartheta
+d_\nu\gamma_1\eta^{\top}\Phi_{\mathcal{J}}\eta,
\qquad s\in\{0,1\}.
\label{eq:branch_bound}
\end{equation}
The $\kappa$ term is dropped because it vanishes when $s=0$ and cannot
contribute a uniform decay rate. Defining the aggregate disturbance
$d\triangleq\gamma_1 d_\varepsilon-\gamma_2\ths$ and substituting
\eqref{eq:branch_bound} into \eqref{eq:vdot_1} yields the second
bound,
\begin{align}
\dot{V}\le{}&-k_1\norm{r}^{2}-\alpha_1\norm{e}^{2}
  -\beta\mathcal{N}\norm{\eta}^{2}
  -2\eta^{\top}\mathcal{J}^{\top}r
  +r^{\top}R\notag\\
&+r^{\top}\varepsilon-c_\nu\gamma_1\vartheta^{\top}\Phi_{\mathcal{J}}\vartheta
 +d_\nu\gamma_1\eta^{\top}\Phi_{\mathcal{J}}\eta
 -\gamma_2\norm{\vartheta}^{2}\notag\\
&+\gamma_2\norm{\eta}^{2}-\gamma_1(\vartheta+\eta)^{\top}d_R
 -(\vartheta+\eta)^{\top}d .
\label{eq:vdot_2}
\end{align}
The cross term is separated by Young's inequality, using the induced norm, as
$-2\eta^{\top}\mathcal{J}^{\top}r\le
\beta\mu_1\norm{\mathcal{J}}_{2}^{2}\norm{\eta}^{2}
+\norm{r}^{2}/(\beta\mu_1)$. Substituting
\eqref{eq:norm_factor} and using
$\eta^{\top}\Phi_{\mathcal{J}}\eta\le
\operatorname{tr}(\Phi_{\mathcal{J}})\norm{\eta}^{2}$, the $\eta$
channel becomes
\begin{align}
&-\beta\mathcal{N}\norm{\eta}^{2}
 +\beta\mu_1\norm{\mathcal{J}}_{2}^{2}\norm{\eta}^{2}
 +d_\nu\gamma_1\eta^{\top}\Phi_{\mathcal{J}}\eta
 +\gamma_2\norm{\eta}^{2}\notag\\
&\qquad\le-B\norm{\eta}^{2}
 -\big(\beta\mu_2-d_\nu\gamma_1\big)
   \operatorname{tr}(\Phi_{\mathcal{J}})\norm{\eta}^{2}
 \le-B\norm{\eta}^{2},
\label{eq:eta_channel}
\end{align}
where the second and third conditions in
\eqref{eq:gain_conditions} were used. The $\vartheta$ channel satisfies
$-c_\nu\gamma_1\vartheta^{\top}\Phi_{\mathcal{J}}\vartheta
-\gamma_2\norm{\vartheta}^{2}\le-a\norm{\vartheta}^{2}$, with
$a=A_0$ for all $t\ge t_0$ because $\Phi_{\mathcal{J}}\succeq0$, and
with $a=A$ for $t\ge\mathcal{T}$ under Assumption~\ref{ass:fe}. The
first condition in \eqref{eq:gain_conditions} gives
$\norm{r}^{2}/(\beta\mu_1)\le(k_1/2)\norm{r}^{2}$. The third bound is
therefore
\begin{align}
\dot{V}\le{}&-\tfrac{k_1}{2}\norm{r}^{2}-\alpha_1\norm{e}^{2}
 -a\norm{\vartheta}^{2}-B\norm{\eta}^{2}
 +r^{\top}R+r^{\top}\varepsilon\notag\\
&-\gamma_1(\vartheta+\eta)^{\top}d_R
 -(\vartheta+\eta)^{\top}d .
\label{eq:vdot_3}
\end{align}

The remaining sign-indefinite terms are bounded by Young's inequality
in the form $ab\le\epsilon a^{2}+b^{2}/(4\epsilon)$, using
\eqref{eq:dR_bound}, \eqref{eq:deps_bound},
\eqref{eq:R_bound_domain}, and $\norm{d}\le D_0$, as
\begin{align}
r^{\top}\varepsilon&\le\tfrac{k_1}{8}\norm{r}^{2}
  +\tfrac{2\bar{\varepsilon}^{2}}{k_1},
r^{\top}R\le\tfrac{k_1}{8}\norm{r}^{2}
  +\tfrac{2\bar{\rho}_0(\chi)^{2}}{k_1}\norm{\tht}^{4},\notag\\
-\gamma_1\vartheta^{\top}d_R&\le\tfrac{a}{4}\norm{\vartheta}^{2}
  +\tfrac{\gamma_1^{2}C_R^{2}}{a}\norm{\tht}^{4},
-\vartheta^{\top}d\le\tfrac{a}{4}\norm{\vartheta}^{2}
  +\tfrac{D_0^{2}}{a},\notag\\
-\gamma_1\eta^{\top}d_R&\le\tfrac{B}{4}\norm{\eta}^{2}
  +\tfrac{\gamma_1^{2}C_R^{2}}{B}\norm{\tht}^{4},
-\eta^{\top}d\le\tfrac{B}{4}\norm{\eta}^{2}
  +\tfrac{D_0^{2}}{B}.
\label{eq:young_terms}
\end{align}
Substituting \eqref{eq:young_terms} into \eqref{eq:vdot_3} and
collecting the coefficients of $\norm{r}^{2}$,
$\norm{\vartheta}^{2}$ and $\norm{\eta}^{2}$ with
\eqref{eq:rho_q} and \eqref{eq:iota} gives the fourth bound,
\begin{align}
\dot{V}\le{}&-\tfrac{k_1}{4}\norm{r}^{2}-\alpha_1\norm{e}^{2}
 -\tfrac{a}{2}\norm{\vartheta}^{2}-\tfrac{B}{2}\norm{\eta}^{2}\notag\\
&+\rho_q(a)\norm{\tht}^{4}+\iota(a).
\label{eq:vdot_4}
\end{align}
Applying
$\norm{\tht}^{4}=\norm{\vartheta-\eta}^{4}\le4\norm{\zeta}^{4}$
reduces \eqref{eq:vdot_4} to
\begin{equation}
\dot{V}\le-\big(\lambda(a)-4\rho_q(a)\norm{\zeta}^{2}\big)\norm{\zeta}^{2}
 +\iota(a).
\label{eq:vdot_collected}
\end{equation}

On $[t_0,t_e)$, where $\norm{\zeta}<\chi$, taking $a=A_0$ and using
the domain condition $4\rho_q(A_0)\chi^{2}\le\lambda(A_0)-\lambda_d$
reduces \eqref{eq:vdot_collected} to
$\dot{V}\le-\lambda_d\norm{\zeta}^{2}+\iota(A_0)
\le-(\lambda_d/\beta_2)V+\iota(A_0)$, and the comparison lemma yields
\begin{equation}
V(t)\le V(t_0)e^{-\frac{\lambda_d}{\beta_2}(t-t_0)}
 +\frac{\beta_2\iota(A_0)}{\lambda_d}
  \Big(1-e^{-\frac{\lambda_d}{\beta_2}(t-t_0)}\Big)
\label{eq:comparison}
\end{equation}
for $t\in[t_0,t_e)$. Since $\zeta(t_0)\in\mathcal{S}$ implies
$V(t_0)<\beta_1\chi^{2}$ by \eqref{eq:lyap_sandwich}, and
$\beta_2\iota(A_0)/\lambda_d<\beta_1\chi^{2}$ by hypothesis, the right
side of \eqref{eq:comparison} is a convex combination of two numbers
below $\beta_1\chi^{2}$ and therefore remains strictly below
$\beta_1\chi^{2}$ on $[t_0,t_e)$, and by continuity also at $t_e$ if
$t_e<\infty$. If $t_e<\infty$, then $\norm{\zeta(t_e)}=\chi$ by
continuity of $\zeta$, so \eqref{eq:lyap_sandwich} gives
$V(t_e)\ge\beta_1\chi^{2}$, a contradiction. Hence $t_e=\infty$ and
$\zeta(t)\in\mathcal{D}$ for all $t\ge t_0$. Consequently
Lemma~\ref{lem:data} covers every retained sample,
Lemma~\ref{lem:observer} gives $\norm{z(t)}\le\delta_\Delta$ for all
$t\ge t_\Delta$, and applying \eqref{eq:lyap_sandwich} to
\eqref{eq:comparison} yields \eqref{eq:main_bound_pre}. For
$t\ge\mathcal{T}$, taking $a=A$ in \eqref{eq:vdot_collected} and
using $\norm{\zeta}<\chi$ gives
$\dot{V}\le-\lambda_{\mathcal{T}}\norm{\zeta}^{2}+\iota(A)
\le-(\lambda_{\mathcal{T}}/\beta_2)V+\iota(A)$, and the comparison
lemma with \eqref{eq:lyap_sandwich} yields \eqref{eq:main_bound}.
\end{proof}

\begin{remark}\label{rem:leakage}
The leakage term $-\gamma_2\hat{\theta}$ is placed in the $\Lambda$ dynamics
rather than in the $\hat{\theta}$ dynamics. Since the cross terms cancel
exactly, this placement produces no mixed $\vartheta^{\top}\eta$ term in the
Lyapunov derivative. Its cost is a positive $\gamma_2\norm{\eta}^{2}$
contribution, which is dominated by the gain condition $\beta>\gamma_2$, and
its benefit is a negative $\gamma_2\norm{\vartheta}^{2}$ contribution that
provides damping in the $\vartheta$ channel even in the absence of
excitation; this damping is the constant $A_0$ on which the pre-excitation
invariance argument in Theorem~\ref{thm:main} is based.
\end{remark}

\section{Simulation}\label{sec:sim}

\subsection{Setup}
The developed method is evaluated on the six-degree-of-freedom
unmanned underwater vehicle of \cite{Patil2024Simultaneous}, with the
mass, inertia and hydrodynamic damping of \cite{Fischer2014UUV}. The
body-fixed inertia is
$M_b=\mathrm{diag}(215,265,265,40,80,80)$, the linear and quadratic
damping coefficients are $\mathrm{diag}(70,100,200,30,50,50)$ and
$\mathrm{diag}(100,200,50,50,100,100)$, and the vehicle is neutrally
buoyant. The state $x\in\R^{6}$ collects the earth-frame position and
orientation, so $X\in\R^{12}$ and $n=6$, and the input $u$ in
\eqref{eq:sys_dyn} is the commanded earth-frame acceleration, from
which the applied wrench follows through the known inertia matrix. The
desired trajectory is the helix
$x_d(t)=[2\cos(0.5t),\,2\sin(0.5t),\,0.1t,\,0,\,0,\,-0.125t]^{\top}$,
initialized at $x(t_0)=[-0.5,-0.5,-0.5,0,0,0]^{\top}$ and
$\dot{x}(t_0)=0$, and each run is integrated for $30$~s at a fixed
step of $10$~ms; the linearly growing depth and yaw are understood to
be continued by a bounded profile after the simulated interval, so that
the bounds on $x_d$ hold.

The DNN has $k=4$ hidden layers of $12$ neurons with $\tanh$
activations and a bias appended to every layer, so $n_\theta=702$. All
methods use identical initial weights, the same plant, and the same
observer. The control and observer gains are $\alpha_1=5$, $k_1=20$,
$\alpha_2=50$ and $k_\Delta=20$. The gating time $t_\Delta=5$~s was
selected conservatively; the observer error $\norm{\tilde{\Delta}}$ falls
below $0.03$ within $2$~s in every run. The history stack
holds $N=150$ samples. Every second sample after $t_\Delta$ is retained
until the stack is full; afterwards a candidate sample is examined every
$50$~ms and replaces the oldest retained sample only if the rank and
minimum eigenvalue of $\Phi_{\mathcal{J}}$, evaluated at the current
weights, do not decrease. The retained states $X_i$ are never otherwise
altered, and their Jacobians are re-evaluated at $\hat{\theta}(t)$ at
every step. The CL gains are $\gamma_1=20$ and $\gamma_2=10^{-3}$, the
gain law \eqref{eq:gain_law} uses $\kappa=0.1$ and $\nu=0.05$ with
$\Gamma(t_0)=0.5I$ confined to the band $[10^{-4},10]$, the higher-order
gain is $\beta=0.5$, and $\mu_1$ and $\mu_2$ are set to $1.05$ times the
lower bounds in \eqref{eq:gain_conditions}. The projection bound is
$\bar{\theta}=100$.

Three methods are compared. The Lb-CL-DNN of \cite{Hart2025Concurrent}
is a first-order CL law using the same history stack, replacement rule,
and observer, with its gain law run at $\nu=1$ and a forgetting factor
of $0.1$. The A-DNN of \cite{Le2025Accelerated} is a higher-order tuner
without CL, run with its own control law and the gains $\gamma_\nu=0.5$,
$\gamma_\theta=4$ and $\beta_s=0.05$ reported therein. The Lb-DNN
baseline is the first-order gradient law of \cite{patil2022lyapunov}
with constant gain $0.5I$. Performance is quantified over $t\in[10,30]$~s
by the RMS of the tracking error $\norm{e}$ and of the function
approximation error $\tilde{f}\triangleq f(X)-\Psi(X,\hat{\theta})$,
each split into linear and angular components, together with
$\tilde{f}$ evaluated at $100$ test inputs obtained by perturbing
reference states at random instants by up to $0.25$ in every coordinate,
which measures whether the function was learned rather than the path
memorized.

\subsection{Discussion}
Table~\ref{tab:results} and Fig.~\ref{fig:approx} summarize the results.
Compared to the first-order CL law of \cite{Hart2025Concurrent}, which is
driven by the same recorded data, the developed method reduces the composite
RMS approximation error from $0.458$ to $0.346$, i.e., by $24.5\%$, with
$22.9\%$ and $26.1\%$ reductions in the linear and angular channels and a
$24.5\%$ reduction in the tracking error. The off-trajectory errors of the
two CL methods are within $2\%$ of each other. Hence, the improvement of the developed law over
\cite{Hart2025Concurrent} lies in the accuracy of learning along the
recorded data rather than in generalization beyond it; since the two laws
also differ in the covariance weighting $\nu$, the comparison does not
attribute the improvement to the momentum channel alone.

The comparison with the higher-order tuner of \cite{Le2025Accelerated}
indicates the role of the recorded data. Since that tuner has no history
stack, it attains $0.933$, which is worse than the $0.777$ of the
first-order gradient baseline, whereas the developed method attains
$0.346$, an improvement of $62.9\%$ over \cite{Le2025Accelerated} and of
$44.4\%$ in the off-trajectory metric. Therefore, momentum alone does not
improve learning on this plant; the higher-order channel must be driven by
the recorded data. Across all runs, $\norm{\hat{\theta}}$ remained below
$8.1$, well inside the projection bound.

At the final time, the history stack $\Phi_{\mathcal{J}}$ evaluated at the
current weights has $\lambda_{\max}=26.9$ and a minimum eigenvalue at machine
precision, with a numerical rank of $182$ out of $702$ at a relative tolerance
of $10^{-9}$. The helix therefore excites only part of the parameter space,
Assumption~\ref{ass:fe} is not met on this trajectory, and the learning
observed in Table~\ref{tab:results} takes place along the excited directions
of $\Phi_{\mathcal{J}}$.
The domain and residual conditions of Theorem~\ref{thm:main} are likewise
sufficient conditions that were not verified at the simulated
gains.\footnote{As is typical of Lyapunov-based robust and adaptive
designs, the gain conditions are sufficient rather than necessary and
are conservative; the gains that satisfy them are not necessarily the
gains that perform best, and the simulated gains were selected for
performance.}

\begin{table}[!t]
\caption{RMS errors over $t\in[10,30]$~s with the retained history
stack, $N=150$, and improvement of the developed method in percent.}
\label{tab:results}
\centering\scriptsize
\newcommand{\numr}[1]{\num[round-mode=places,round-precision=1]{#1}}
\setlength{\tabcolsep}{1pt}
\begin{tabular}{lcccccc}
\toprule
Method & $\|e_{\mathrm{lin}}\|$ & $\|e_{\mathrm{ang}}\|$
 & $\|\tilde{f}_{\mathrm{lin}}\|$ & $\|\tilde{f}_{\mathrm{ang}}\|$
 & $\|\tilde{f}_{\mathrm{ot}}\|$ & $\|\tilde{f}\|$\\
\midrule
\csvreader[filter strcmp={\csvcoli}{cmp1}, late after line=\\]%
{metrics1.csv}%
{4=\lab,5=\elin,6=\eang,9=\ftlin,10=\ftang,11=\ftot,14=\ft}%
{\lab & \num{\elin} & \num{\eang} & \num{\ftlin} & \num{\ftang}
 & \num{\ftot} & \num{\ft}}
\csvreader[filter expr={test{\ifcsvstrequal{\csvcoli}{cmp2}}
 and test{\ifcsvstrequal{\csvcoliii}{accel}}}, late after line=\\]%
{metrics1.csv}%
{4=\lab,5=\elin,6=\eang,9=\ftlin,10=\ftang,11=\ftot,14=\ft}%
{\lab & \num{\elin} & \num{\eang} & \num{\ftlin} & \num{\ftang}
 & \num{\ftot} & \num{\ft}}
\midrule
\csvreader[filter strcmp={\csvcoli}{cmp1_impr}, late after line=\\]%
{metrics1.csv}%
{4=\lab,5=\elin,6=\eang,9=\ftlin,10=\ftang,11=\ftot,14=\ft}%
{Impr.\ vs.\ \lab~[\%] & \numr{\elin} & \numr{\eang} & \numr{\ftlin}
 & \numr{\ftang} & \numr{\ftot} & \numr{\ft}}
\csvreader[filter expr={test{\ifcsvstrequal{\csvcoli}{cmp2_impr}}
 and test{\ifcsvstrequal{\csvcoliii}{accel}}}, late after line=\\]%
{metrics1.csv}%
{4=\lab,5=\elin,6=\eang,9=\ftlin,10=\ftang,11=\ftot,14=\ft}%
{Impr.\ vs.\ \lab~[\%] & \numr{\elin} & \numr{\eang} & \numr{\ftlin}
 & \numr{\ftang} & \numr{\ftot} & \numr{\ft}}
\bottomrule
\end{tabular}
\end{table}

\begin{figure}[!t]
\centering
\includegraphics[width=\columnwidth]{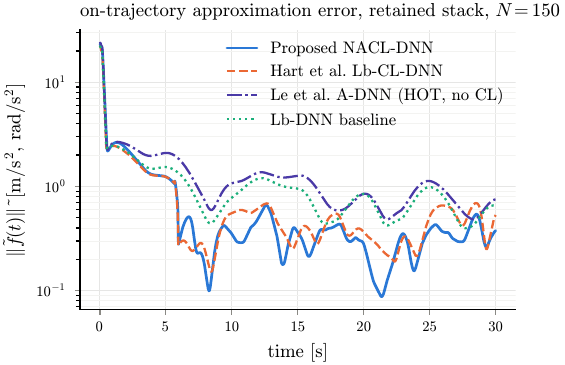}
\caption{Function approximation error $\|\tilde{f}(t)\|$ of the
developed method against the first-order CL law of
\cite{Hart2025Concurrent}, the higher-order tuner of
\cite{Le2025Accelerated} without a history stack, and the first-order
Lb-DNN baseline.}
\label{fig:approx}
\end{figure}

\section{Conclusion}\label{sec:conclusion}

A Nesterov-accelerated concurrent learning (NACL) adaptation law was developed for
a Lyapunov-based DNN controller that updates all layers online, in which the
momentum variable is driven by a finitely exciting history stack and the
higher-order cross term is dominated by a state-dependent normalization. The
Lyapunov-based analysis established exponential convergence of the tracking,
momentum, weight estimation, and observer errors to a ball whose radius is
set by the approximation error, the observer accuracy, and the leakage gain,
from initialization and under a finite excitation condition on the retained
Jacobian data rather than persistent excitation of the trajectory. On a
six-degree-of-freedom UUV, the developed method reduced the RMS function
approximation error by $24.5\%$ relative to the first-order CL law of
\cite{Hart2025Concurrent} driven by the same data and by $62.9\%$ relative
to the higher-order tuner without CL of \cite{Le2025Accelerated}.

\section*{Acknowledgment}
AI-based tools were used for the purposes of simulation code development, proofreading, and language review. The core mathematical
development, the stability analysis, and the conceptual framework were
entirely performed by the authors. All AI-generated outputs and suggestions
were thoroughly reviewed, verified, and approved by the authors, who take
full responsibility for the content of the final manuscript.

\bibliographystyle{IEEEtran}\bibliography{ref}

\begin{thebibliography}{10}
\providecommand{\url}[1]{#1}
\csname url@samestyle\endcsname
\providecommand{\newblock}{\relax}
\providecommand{\bibinfo}[2]{#2}
\providecommand{\BIBentrySTDinterwordspacing}{\spaceskip=0pt\relax}
\providecommand{\BIBentryALTinterwordstretchfactor}{4}
\providecommand{\BIBentryALTinterwordspacing}{\spaceskip=\fontdimen2\font plus
\BIBentryALTinterwordstretchfactor\fontdimen3\font minus
  \fontdimen4\font\relax}
\providecommand{\BIBforeignlanguage}[2]{{%
\expandafter\ifx\csname l@#1\endcsname\relax
\typeout{** WARNING: IEEEtran.bst: No hyphenation pattern has been}%
\typeout{** loaded for the language `#1'. Using the pattern for}%
\typeout{** the default language instead.}%
\else
\language=\csname l@#1\endcsname
\fi
#2}}
\providecommand{\BIBdecl}{\relax}
\BIBdecl

\bibitem{Hornik1989Multilayer}
K.~Hornik, M.~Stinchcombe, and H.~White, ``Multilayer feedforward networks are
  universal approximators,'' \emph{Neural Netw.}, vol.~2, no.~5, pp. 359--366,
  1989.

\bibitem{Narendra1990Identification}
K.~S. Narendra and K.~Parthasarathy, ``Identification and control of dynamical
  systems using neural networks,'' \emph{IEEE Trans. Neural Netw.}, vol.~1,
  no.~1, pp. 4--27, Mar. 1990.

\bibitem{Lewis1996Neural}
F.~L. Lewis, A.~Yesildirek, and K.~Liu, ``Multilayer neural-net robot
  controller with guaranteed tracking performance,'' \emph{IEEE Trans. Neural
  Netw.}, vol.~7, no.~2, pp. 388--399, Mar. 1996.

\bibitem{Sanner1992Gaussian}
R.~M. Sanner and J.-J.~E. Slotine, ``{Gaussian} networks for direct adaptive
  control,'' \emph{IEEE Trans. Neural Netw.}, vol.~3, no.~6, pp. 837--863, Nov.
  1992.

\bibitem{Ge2002Stable}
S.~S. Ge, C.~C. Hang, T.~H. Lee, and T.~Zhang, \emph{Stable Adaptive Neural
  Network Control}.\hskip 1em plus 0.5em minus 0.4em\relax Boston, MA, USA:
  Kluwer, 2002.

\bibitem{Farrell2006Poly}
J.~A. Farrell and M.~M. Polycarpou, \emph{Adaptive Approximation Based Control:
  Unifying Neural, Fuzzy and Traditional Adaptive Approximation
  Approaches}.\hskip 1em plus 0.5em minus 0.4em\relax Hoboken, NJ, USA: Wiley,
  2006.

\bibitem{Lavretsky2012Robust}
E.~Lavretsky and K.~A. Wise, \emph{Robust and Adaptive Control: With Aerospace
  Applications}.\hskip 1em plus 0.5em minus 0.4em\relax London, U.K.: Springer,
  2012.

\bibitem{patil2022lyapunov}
O.~S. Patil, D.~M. Le, M.~L. Greene, and W.~E. Dixon, ``{Lyapunov}-derived
  control and adaptive update laws for inner and outer layer weights of a deep
  neural network,'' \emph{IEEE Control Syst. Lett.}, vol.~6, pp. 1855--1860,
  2022.

\bibitem{joshi2020design}
G.~Joshi, J.~Virdi, and G.~Chowdhary, ``Design and flight evaluation of deep
  model reference adaptive controller,'' in \emph{Proc. AIAA SciTech Forum},
  2020, {AIAA} Paper 2020-1336.

\bibitem{sun2022lyapunov}
R.~Sun, M.~L. Greene, D.~M. Le, Z.~I. Bell, G.~Chowdhary, and W.~E. Dixon,
  ``{Lyapunov}-based real-time and iterative adjustment of deep neural
  networks,'' \emph{IEEE Control Syst. Lett.}, vol.~6, pp. 193--198, 2022.

\bibitem{griffis2023lyapunov}
E.~J. Griffis, O.~S. Patil, Z.~I. Bell, and W.~E. Dixon, ``{Lyapunov}-based
  long short-term memory ({Lb-LSTM}) neural network-based control,'' \emph{IEEE
  Control Syst. Lett.}, vol.~7, pp. 2976--2981, 2023.

\bibitem{hart2024lyapunov}
R.~G. Hart, E.~J. Griffis, O.~S. Patil, and W.~E. Dixon, ``{Lyapunov}-based
  physics-informed long short-term memory ({LSTM}) neural network-based
  adaptive control,'' \emph{IEEE Control Syst. Lett.}, vol.~8, pp. 13--18,
  2024.

\bibitem{makumi2025lyapunov}
W.~A. Makumi, O.~S. Patil, and W.~E. Dixon, ``{Lyapunov}-based adaptive deep
  system identification for approximate dynamic programming,''
  \emph{Automatica}, vol. 180, p. 112462, Oct. 2025.

\bibitem{Patil2024Simultaneous}
O.~S. Patil, E.~J. Griffis, W.~A. Makumi, and W.~E. Dixon, ``Simultaneous
  online system identification and control using composite adaptive
  {Lyapunov}-based deep neural networks,'' 2026, arXiv:2311.13056.

\bibitem{Ioannou1996Robust}
P.~A. Ioannou and J.~Sun, \emph{Robust Adaptive Control}.\hskip 1em plus 0.5em
  minus 0.4em\relax Upper Saddle River, NJ, USA: Prentice-Hall, 1996.

\bibitem{Narendra1989Stable}
K.~S. Narendra and A.~M. Annaswamy, \emph{Stable Adaptive Systems}.\hskip 1em
  plus 0.5em minus 0.4em\relax Englewood Cliffs, NJ, USA: Prentice-Hall, 1989.

\bibitem{Astrom2008Adaptive}
K.~J. {\AA}str{\"o}m and B.~Wittenmark, \emph{Adaptive Control}, 2nd~ed.\hskip
  1em plus 0.5em minus 0.4em\relax Mineola, NY, USA: Dover, 2008.

\bibitem{Sastry1989Adaptive}
S.~Sastry and M.~Bodson, \emph{Adaptive Control: Stability, Convergence, and
  Robustness}.\hskip 1em plus 0.5em minus 0.4em\relax Englewood Cliffs, NJ,
  USA: Prentice-Hall, 1989.

\bibitem{Aranovskiy2017DREM}
S.~Aranovskiy, A.~Bobtsov, R.~Ortega, and A.~Pyrkin, ``Performance enhancement
  of parameter estimators via dynamic regressor extension and mixing,''
  \emph{IEEE Trans. Autom. Control}, vol.~62, no.~7, pp. 3546--3550, 2017.

\bibitem{ortega2021new}
R.~Ortega, S.~Aranovskiy, A.~A. Pyrkin, A.~Astolfi, and A.~A. Bobtsov, ``New
  results on parameter estimation via dynamic regressor extension and mixing:
  Continuous and discrete-time cases,'' \emph{IEEE Trans. Autom. Control},
  vol.~66, no.~5, pp. 2265--2272, May 2021.

\bibitem{bobtsov2022generation}
A.~Bobtsov, B.~Yi, R.~Ortega, and A.~Astolfi, ``Generation of new exciting
  regressors for consistent online estimation of unknown constant parameters,''
  \emph{IEEE Trans. Autom. Control}, vol.~67, no.~9, pp. 4746--4753, 2022.

\bibitem{jenkins2020fast}
B.~Jenkins, A.~Krupadanam, and A.~M. Annaswamy, ``Fast adaptive observers for
  battery management systems,'' \emph{IEEE Trans. Control Syst. Technol.},
  vol.~28, no.~3, pp. 776--789, 2020.

\bibitem{kreisselmeier1977adaptive}
G.~Kreisselmeier, ``Adaptive observers with exponential rate of convergence,''
  \emph{IEEE Trans. Autom. Control}, vol.~22, no.~1, pp. 2--8, 1977.

\bibitem{roy2018combined}
S.~B. Roy, S.~Bhasin, and I.~N. Kar, ``Combined {MRAC} for unknown {MIMO} {LTI}
  systems with parameter convergence,'' \emph{IEEE Trans. Autom. Control},
  vol.~63, no.~1, pp. 283--290, 2018.

\bibitem{katiyar2023initial}
A.~Katiyar, S.~B. Roy, and S.~Bhasin, ``Initial-excitation-based robust
  adaptive observer for {MIMO} {LTI} systems,'' \emph{IEEE Trans. Autom.
  Control}, vol.~68, no.~4, pp. 2536--2543, 2023.

\bibitem{cho2018composite}
N.~Cho, H.-S. Shin, Y.~Kim, and A.~Tsourdos, ``Composite model reference
  adaptive control with parameter convergence under finite excitation,''
  \emph{IEEE Trans. Autom. Control}, vol.~63, no.~3, pp. 811--818, 2018.

\bibitem{chowdhary2013concurrent}
G.~Chowdhary, T.~Yucelen, M.~M{\"u}hlegg, and E.~N. Johnson, ``Concurrent
  learning adaptive control of linear systems with exponentially convergent
  bounds,'' \emph{Int. J. Adapt. Control Signal Process.}, vol.~27, no.~4, pp.
  280--301, 2013.

\bibitem{kamalapurkar2017concurrent}
R.~Kamalapurkar, B.~Reish, G.~Chowdhary, and W.~E. Dixon, ``Concurrent learning
  for parameter estimation using dynamic state-derivative estimators,''
  \emph{IEEE Trans. Autom. Control}, vol.~62, no.~7, pp. 3594--3601, Jul. 2017.

\bibitem{chowdhary2013cladmpc}
G.~Chowdhary, M.~M{\"u}hlegg, J.~P. How, and F.~Holzapfel, ``Concurrent
  learning adaptive model predictive control,'' in \emph{Advances in Aerospace
  Guidance, Navigation and Control}.\hskip 1em plus 0.5em minus 0.4em\relax
  Berlin, Germany: Springer, 2013, pp. 29--47.

\bibitem{parikh2019integral}
A.~Parikh, R.~Kamalapurkar, and W.~E. Dixon, ``Integral concurrent learning:
  Adaptive control with parameter convergence using finite excitation,''
  \emph{Int. J. Adapt. Control Signal Process.}, vol.~33, no.~12, pp.
  1775--1787, Dec. 2019.

\bibitem{slotine1989composite}
J.-J.~E. Slotine and W.~Li, ``Composite adaptive control of robot
  manipulators,'' \emph{Automatica}, vol.~25, no.~4, pp. 509--519, Jul. 1989.

\bibitem{pan2018composite}
Y.~Pan and H.~Yu, ``Composite learning robot control with guaranteed parameter
  convergence,'' \emph{Automatica}, vol.~89, pp. 398--406, Mar. 2018.

\bibitem{Basyal2024Augmentation}
S.~Basyal, J.~Ting, K.~Mishra, and B.~C. Allen, ``Augmentation of a
  {Lyapunov}-based deep neural network controller with concurrent learning for
  control-affine nonlinear systems,'' in \emph{Proc. Amer. Control Conf.},
  2024, pp. 2885--2890.

\bibitem{Basyal2024EulerLagrange}
------, ``Concurrent learning and {Lyapunov}-based updates of deep neural
  networks for {Euler}-{Lagrange} dynamic systems,'' in \emph{Proc. Amer.
  Control Conf.}, 2024, pp. 15--20.

\bibitem{Basyal2026Realtime}
S.~Basyal, J.~Ting, K.~Mishra, and B.~Allen, ``A lyapunov-and concurrent
  learning-based approach to train the weights of a dnn-based controller in
  real-time,'' in \emph{2026 American Control Conference (ACC)}, 2026, pp.
  444--449.

\bibitem{Hart2025Concurrent}
R.~G. Hart, O.~S. Patil, Z.~I. Bell, and W.~E. Dixon, ``Concurrent learning for
  system identification and control using {Lyapunov}-based deep neural
  networks,'' \emph{IEEE Control Syst. Lett.}, vol.~9, pp. 2957--2962, 2025.

\bibitem{Nesterov1983Method}
Y.~E. Nesterov, ``A method of solving a convex programming problem with
  convergence rate {$O(1/k^2)$},'' \emph{Sov. Math. Dokl.}, vol.~27, no.~2, pp.
  372--376, 1983.

\bibitem{Su2014Differential}
W.~Su, S.~Boyd, and E.~J. Cand{\`e}s, ``A differential equation for modeling
  {Nesterov}'s accelerated gradient method: Theory and insights,'' in
  \emph{Adv. Neural Inf. Process. Syst.}, vol.~27, 2014, pp. 2510--2518.

\bibitem{gaudio2019connections}
J.~E. Gaudio, T.~E. Gibson, A.~M. Annaswamy, M.~A. Bolender, and E.~Lavretsky,
  ``Connections between adaptive control and optimization in machine
  learning,'' in \emph{Proc. IEEE Conf. Decis. Control}, 2019, pp. 4563--4568.

\bibitem{gaudio2021high}
J.~E. Gaudio, A.~M. Annaswamy, M.~A. Bolender, E.~Lavretsky, and T.~E. Gibson,
  ``A class of high order tuners for adaptive systems,'' \emph{IEEE Control
  Syst. Lett.}, vol.~5, no.~2, pp. 391--396, 2021.

\bibitem{cui2022accelerated}
Y.~Cui and A.~M. Annaswamy, ``Accelerated performance and accelerated learning
  with discrete-time high-order tuners,'' in \emph{Proc. IEEE Conf. Decis.
  Control}, 2022, pp. 1350--1355.

\bibitem{gaudio2022parameter}
J.~E. Gaudio, A.~M. Annaswamy, E.~Lavretsky, and M.~A. Bolender, ``Parameter
  estimation in adaptive control of time-varying systems under a range of
  excitation conditions,'' \emph{IEEE Trans. Autom. Control}, vol.~67, no.~10,
  pp. 5440--5447, 2022.

\bibitem{BasuRoy2026Exponential}
S.~B. Roy and A.~Annaswamy, ``Exponential stability in the large and
  accelerated convergence for higher order adaptive tuners with persistence of
  excitation,'' \emph{IEEE Trans. Autom. Control}, 2026, early access.

\bibitem{pan2023least}
Y.~Pan, T.~Shi, and C.~Wen, ``Least-squares composite learning backstepping
  control with high-order tuners,'' in \emph{Proc. IEEE Conf. Decis. Control},
  2023, pp. 515--520.

\bibitem{shi2024composite}
T.~Shi, C.~Wen, and Y.~Pan, ``Composite learning backstepping control with
  guaranteed exponential stability and robustness,'' 2024, arXiv:2401.10785.

\bibitem{Le2022Integral}
D.~M. Le, O.~S. Patil, P.~M. Amy, and W.~E. Dixon, ``Integral concurrent
  learning-based accelerated gradient adaptive control of uncertain
  {Euler}-{Lagrange} systems,'' in \emph{Proc. Amer. Control Conf.}, 2022, pp.
  806--811.

\bibitem{Le2025Accelerated}
D.~M. Le, O.~S. Patil, C.~F. Nino, and W.~E. Dixon, ``Accelerated gradient
  approach for deep neural network-based adaptive control of unknown nonlinear
  systems,'' \emph{IEEE Trans. Neural Netw. Learn. Syst.}, vol.~36, no.~4, pp.
  6299--6313, 2025.

\bibitem{Patil2026Bounds}
O.~S. Patil, B.~C. Fallin, C.~F. Nino, R.~G. Hart, and W.~E. Dixon, ``Bounds on
  deep neural network partial derivatives with respect to parameters,'' 2026,
  arXiv:2503.21007.

\bibitem{Krstic1995}
M.~Krsti\'{c}, I.~Kanellakopoulos, and P.~V. Kokotovi\'{c}, \emph{Nonlinear and
  Adaptive Control Design}.\hskip 1em plus 0.5em minus 0.4em\relax New York,
  NY, USA: Wiley, 1995.

\bibitem{Fischer2014UUV}
N.~Fischer, D.~Hughes, P.~Walters, E.~M. Schwartz, and W.~E. Dixon, ``Nonlinear
  {RISE}-based control of an autonomous underwater vehicle,'' \emph{IEEE Trans.
  Robot.}, vol.~30, no.~4, pp. 845--852, 2014.

\end{thebibliography}

\end{document}